\documentclass[conference,a4paper]{IEEEtran} 

\usepackage{epsfig,cite}
\usepackage{graphicx}
\usepackage{color}
\usepackage{amssymb,amsmath}

\IEEEoverridecommandlockouts

\newtheorem{theorem}{Theorem}

\newtheorem{lemma}{Lemma}

\def\given{\:|\:}
\def\L{\mathsf{L}}
\def\H{\mathsf{H}}
\def\U{\mathsf{U}}
\def\B{\mathsf{B}}

\newcommand{\cfb}{C_{\mathrm{fb}}}
\newcommand{\ciid}{C_{\mathrm{iid}}}

\def\Pr{{\mathrm{Pr}}}

\begin{document}

\title{Capacity of a Simple Intercellular Signal Transduction Channel}

\author{
	\authorblockN{Andrew W. Eckford}
	\authorblockA{Dept. of Computer Science and Engineering\\
	York University\\
	Toronto, Ontario, Canada M3J 1P3\\
	Email: aeckford@yorku.ca}
	\and
	\authorblockN{Peter J. Thomas}
	\authorblockA{Dept. of Mathematics and Dept. of Biology\\
	Case Western Reserve University\\
	Cleveland, Ohio, USA 44106-7058\\
	Email: pjthomas@case.edu}%
	\thanks{This work was funded by grants from the National Science Foundation (DMS-0720142, EF-1038677) and the Natural Sciences and Engineering Research Council (NSERC).}%
}

\maketitle

\begin{abstract}
We model the ligand-receptor molecular communication channel with a 
discrete-time Markov model, and show how to obtain the capacity of this channel. 
We show that the capacity-achieving input distribution is iid; further, 
unusually for a channel with memory,
we show that feedback does not increase the capacity of this channel.
\end{abstract}


\section{Introduction}

Microorganisms communicate using
{\em molecular communication}, in which messages are expressed as patterns of molecules,
propagating via diffusion from transmitter to receiver: what can information theory say about 
this communication?
The physics and mathematics of Brownian motion and chemoreception are 
well understood (e.g., \cite{karatzas-book, berg77}), so it is possible
to construct channel models and calculate information-theoretic quantities, such as capacity \cite{NIPS2006}.
We expect that Shannon's channel coding theorem, and other limit theorems in 
information theory, express ultimate limits on reliable communication, not just for 
human-engineered systems, but for naturally occurring systems as well. We can hypothesize that evolutionary pressure may have optimized natural molecular communication
systems with respect to these limits \cite{AgarwalaChielThomas2012JTB}.  Calculating quantities such as
capacity may allow us to make predictions about biological systems, and explain
biological behaviour \cite{CheongRheeWangNemenmanLevchenko2011Science,Thomas:2011:Science}.

Recent work on molecular communication can be divided into two categories.
In the first category, work has focused on the engineering possibilities:
to exploit molecular communication for specialized applications, such as nanoscale networking
\cite{hiy05,par09}. In this direction, information-theoretic work has focused on
the ultimate capacity of these channels, regardless of biological mechanisms (e.g., \cite{eck07,son12}).
In the second category, work has focused on analyzing the biological machinery of
molecular communication (particularly ligand-receptor systems), both to describe the components of a possible communication system
\cite{nak07} and to describe their capacity \cite{ata07,NIPS2003_NS03,ein11,ein11b}. Our
paper, which builds on work presented in \cite{NIPS2003_NS03}, fits into this category, and many tools in the information-theoretic literature can be used to
solve problems of this type. Related work is also found in \cite{ein11}, where capacity-achieving input distributions were 
found for a simplified ``ideal'' receptor; that paper also discusses but does not solve  the capacity
for the channel model we use.

\section{Models}
\label{sec:Model}

{\em Notation.} Capital letters, e.g., $X$, are random variables; lower-case
letters are constants or 
particular values of the corresponding random variable, e.g., $x$ is a 
particular value of $X$. Vectors use superscripts: $X^i$ represents
an $i$-fold random vector with elements $[X_1,X_2,\ldots,X_i]$; $x^i$ represents
a particular value of $X^i$. Script letters, e.g., $\mathcal{X}$, are sets. 
The logarithm is base 2 unless specified.

\subsection{Physical model}
\label{sec:PhysicalModel}

Signalling between biological cells involves the transmission of signalling molecules, or ligands. These ligands propagate through a shared medium until they are absorbed by a receptor on the surface of a cell. Thus, a message can be passed to a cell by affecting the states of the receptors on its surface; moreover, this process can be modelled as a finite state machine. This
setup is depicted in Figure \ref{fig:LigandReceptor}, and our goal in this paper is to calculate the information-theoretic capacity of this channel.

Finite state Markov processes conditional on an input process provide models of signal transduction and communication in a variety of biological systems, including chemosensation via ligand-receptor interaction, \cite{Wang+Rappel+Kerr+Levine:2007:PRE,Ueda+Shibata:2007:BPJ}, 
dynamics of ion channels sensitive to signals carried by voltage, neurotransmitter concentration, or light \cite{Colquhoun+Hawkes:1983chapter,KellerFranksBartolSejnowski2008PLoSOne,NikolicLoizuDegenaarToumazou2010IntegrBiol}.  
Typically a single ion channel or receptor is in one of $n$ states, with instantaneous transition rate matrix $\mathbf{Q} = [q_{jk}]$ depending on an external input $X(t)$. The probability, $p_k$, that the channel is in state $Y(t)=k\in\mathcal{Y}$ evolves according to 
\begin{equation}\label{eq:cts-time-general}
dp_k/dt=\sum_{j=1}^n p_j(t)q_{jk}(X(t))
\end{equation}
where for  $(j\ne k)$, $q_{jk}$ is the input-dependent per capita rate at which the receptor transitions from state $j$ to state $k$, and $q_{jj}=-\sum_{k\ne j}q_{jk}$.
Taking $\{X(t)\}_{t=0}^T$ as the input, and the receptor state $Y(t)\in\mathcal{Y}$ as the output, gives a channel model, the capacity of which is of general interest.

Here we specialize from (\ref{eq:cts-time-general}) to the case of a single receptor that can be in one of two states, either bound to a signaling molecule or ligand ($Y=\B$), or unbound ($Y=\U$) and hence available to bind.  Thus $\mathcal{Y}=\{\U,\B\}$. (In practice, signals are transduced in parallel by multiple receptor protein molecules, however in many instances they act to a good approximation as independent receivers of a common ligand concentration signal, in which case analysis of the single molecule channel can provides a useful reference point.) When the receptor is bound by a ligand molecule, the signal is said to be \emph{transduced};  typically the receptor (a large protein molecule) undergoes a conformational shift upon binding the ligand.  This change then signals the presence of the ligand through a cascade of intracellular reactions catalyzed by the bound receptor.   The ligand-receptor interaction comprises two chemical reactions, a binding reaction (ligand + receptor $\longrightarrow$ bound receptor) with on-rate $k_+$, and a reverse, unbinding reaction with off-rate $k_-$.  In a continuous time model, let $p(t)=\Pr[Y_t=\B]$.  Then (\ref{eq:cts-time-general}) reduces to 
\begin{equation}\label{eq:cts-time-2state}
dp/dt=k_+c(t)(1-p(t)) - k_-p(t),
\end{equation}
where $c(t)$ is the time-varying ligand concentration.  

A key feature distinguishing this channel is that the receptor is insensitive to the input when in the state $Y=\B$, and can only transduce information about the input, $X(t)=k_+c(t)$, when $Y=\U$.  Thus, analysis of the ligand-binding channel is complicated by the receptor's insensitivity to changes
in concentration occurring while the receptor is in the occupied state. 
This fact plays a decisive role in our proof of our main result, which asserts that feedback from the channel state to the input process cannot increase the capacity, in a discrete time analog of this simple model of intercellular communication. 

 In the limit in which 
transition from the bound state back to the unbound state is instantaneous, the ligand-binding channel becomes
a simple counting process, with the input encoded in the time varying intensity.  This situation
is exactly the one considered in Kabanov's analysis of the capacity of a Poisson channel, under a max/min intensity constraint \cite{Davis1980ieeeIT,Kabanov1978}.  For the Poisson channel, the capacity may be achieved by setting the input to be a two-valued random process fluctuating between the maximum and minimum intensities.  If the intensity is restricted to lie in the interval $[1,1+c]$, the capacity is \cite{Kabanov1978}
\begin{equation}\label{eq:Kabanov}
C_{\mbox{Kab}}(c)=\frac{(c+1)^{1+1/c}}{e}-\left(1+\frac{1}{c} \right)\ln(c+1).
\end{equation}
Our long-term goal is to obtain expressions analogous to (\ref{eq:Kabanov}) for the continuous-time systems (\ref{eq:cts-time-general}) and (\ref{eq:cts-time-2state}).  As a first step, we  restrict attention to a discrete time analog of the two-state system (\ref{eq:cts-time-2state}).
Kabanov's formula may be obtained by restricting the input to a two-state discrete time Markov process with input $X(t)$ taking the values $X_{lo}=1$ and $X_{hi}=1+c$, with transitions $X_{lo}\to X_{hi}$ happening with probability $r$, and transitions $X_{hi}\to X_{lo}$ with probability $s$, per time step.  Maximizing the mutual information with respect to $r$ and $s$, and taking the limit of small time steps, yields  (\ref{eq:Kabanov}). 
In addition, Kabanov proved that the capacity of the Poisson channel cannot be increased by allowing feedback.

\subsection{Mathematical model}

Motivated by the preceding discussion, we examine a discrete-time, finite-state Markov representation of both the
transmission process and the observation process.  We also use a two-state Markov chain
to represent the state of the observer.  As in the continuous time case, the receiver may either be in an unbound state, in which the receiver is waiting for a molecule to bind to the receptor, or in a bound state, in which the receiver has captured a molecule, and must release it before capturing another.

\begin{figure}[t!]
	\begin{center}
	\includegraphics[width=3.5in]{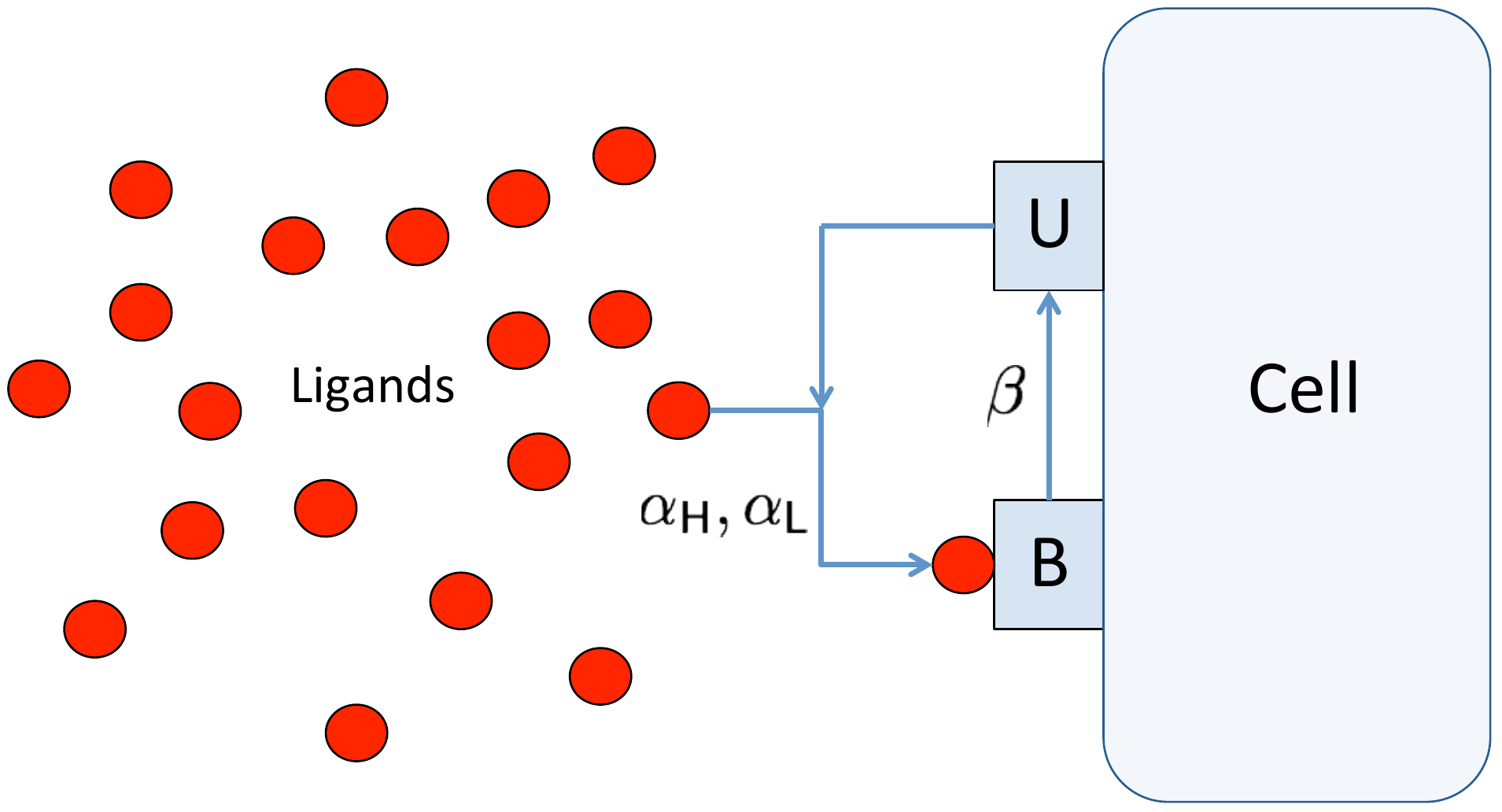}
	\end{center}
	\caption{\label{fig:LigandReceptor} A depiction of our system, in which information is passed to the cell by affecting the state of the receptor. When the ligand binds to the receptor, the receptor enters the $\B$ state, no longer sensitive to the concentration of ligands. When the ligand unbinds (leaves), the cell enters the $\U$ state. While in the $\U$ state, the binding rate is dependent of the concentration, either $\H$ or $\L$. Discrete-time state transition probabilities $\alpha_\H$, $\alpha_\L$, and $\beta$ are illustrated.}
\end{figure}

Let $\mathcal{X} = \{\L,\H\}$ represent the input alphabet, where $\L$ represents low concentration,
and $\H$ represents high concentration. Let $X^n = [X_1, X_2, \ldots, X_n]$  represent a sequence of (random) inputs, where $X_i \in \mathcal{X}$ for all $i$. For now, we make no assumptions on the distribution of $X^n$.
As before, let $\mathcal{Y} = \{\U,\B\}$ represent the output alphabet, where $\U$ represents
the unbound state 
and $\B$ represents the bound state. 
Also, let $Y^n = [Y_1, Y_2, \ldots, Y_n]$ represent a sequence of outputs,
where $Y_i \in \mathcal{Y}$ for all $i$. 

We define parameters 
to bring the continuous-time dynamics, expressed in (\ref{eq:cts-time-general}),
into discrete time.
In our model, the transition probability from $\U$ to $\B$ 
(called the {\em binding rate})
is dependent on the input concentration $x_i$. However, the transition probability
from $\B$ to $\U$ (called the {\em unbinding rate}) is independent of $x_i$.
Thus, given $x^n$, $y^n$ forms a nonstationary Markov chain with
three parameters:
\begin{itemize}
	\item $\alpha_\L$, the binding rate given $x_i = \L$;
	\item $\alpha_\H$, the binding rate given $x_i = \H$; and 
	\item $\beta$, the unbinding rate (independent of $x_i$).
\end{itemize}
We assume $\alpha_\H \geq \alpha_\L$, since binding is more likely at high concentration.

If $x_i = \L$, then
the transition probability matrix is given by
\begin{equation}
	\label{eqn:LowConcentrationOutput}
	\mathbf{P}_{Y|X=\L} = 
	\left[
		\begin{array}{cc}
			1 - \alpha_\L & \alpha_\L \\
			\beta & 1 - \beta
		\end{array}
	\right] ,
\end{equation}
with entries for $\U$ on the first row and column, and $\B$ on the second row and column.
If $x_i = \H$, we have
\begin{equation}
	\label{eqn:HighConcentrationOutput}
	\mathbf{P}_{Y|X=\H} = 
	\left[
		\begin{array}{cc}
			1 - \alpha_\H & \alpha_\H \\
			\beta & 1 - \beta
		\end{array}
	\right] .
\end{equation}
These matrices thus specify  $p_{Y_{n+1} | X_{n+1},Y_n}(y_{n+1} | x_{n+1},y_n)$.

\section{Capacity of the intercellular transduction channel}

The main result of this paper is to show that capacity of the 
discrete-time intercellular
transduction channel is achieved by an iid 
input distribution for $0 < \alpha_\L, \alpha_\H, \beta < 1$.
Our approach is to start with the feedback capacity, show that it is
achieved with an iid input distribution, and conclude that feedback capacity
must therefore be equal to regular capacity; unusually for a channel
with memory, feedback does not increase capacity of our channel. In proving these 
statements, we rely 
on the important results on feedback capacity from \cite{yin03,che05}.

Let $C$ represent the capacity of the system without feedback,
and let $C_{\mathrm{iid}}$ represent the capacity of the system, restricting the
input distribution to be iid. Then the main result is formally stated as follows:
\begin{theorem}
	\label{thm:main}
	For the intercellular signal transduction channel
	described in this paper, if $0 < \alpha_\L, \alpha_\H, \beta < 1$,
	\begin{equation}
		\cfb = C = \ciid .
	\end{equation}
\end{theorem}
The remainder of this section is dedicated to the proof of Theorem \ref{thm:main}.

We start with feedback capacity, which is defined using directed information.
The directed information between vectors $X^n$ and $Y^n$ \cite{Massey1990ieee-conf-IT} 
is given by
\begin{equation}
	\label{eqn:DirectedInfo}
	I(X^n \rightarrow Y^n) = \sum_{i=1}^n I(X^i;Y_i \given Y^{i-1}) .
\end{equation}
The per-symbol directed information rate is given by
\begin{equation}
	\lim_{n \rightarrow \infty} \frac{1}{n} I(X^n \rightarrow Y^n) .
\end{equation}
Feedback capacity, $\cfb$, is then given by
\begin{equation}
	\label{eqn:FeedbackCapacity}
	\cfb = \max_{p_{X^n|Y^n}(x^n\given y^n) \in \mathcal{P}}
	\left(
		\lim_{n \rightarrow \infty} \frac{1}{n} I(X^n \rightarrow Y^n) 
	\right) ,
\end{equation}
where $\mathcal{P}$ represents the set of causal-conditional feedback input distributions:
$p_{X^n|Y^n}(x^n\given y^n) \in \mathcal{P}$ if and only if
$p_{X^n|Y^n}(x^n\given y^n)$ can be written as
\begin{equation}
	p_{X^n|Y^n}(x^n\given y^n) = \prod_{k=2}^n 
	p_{X_k|X^{k-1},Y^{k-1}}(x_k \given x^{k-1}, y^{k-1})p_{X_1}(x_1) .
\end{equation}
%
%

Let $\mathcal{P}^* \subseteq \mathcal{P}$ represent the set of feedback input distributions that 
can be written
\begin{equation}
	p_{X^n|Y^n}(x^n\given y^n) = \prod_{i=2}^n 
	p_{X_i|Y_{i-1}}(x_i \given y_{i-1}) p_{X_1}(x_1).
\end{equation}
(Note that distributions in $\mathcal{P}^*$ need not be stationary:
$p_{X_i|Y_{i-1}}(x \given y)$ can depend on $i$.)
Then $\mathcal{P}^* \subset \mathcal{P}$ for $n>2$. The following result,
found in the literature, says there is at least one feedback-capacity-achieving
input distribution in $\mathcal{P}^*$.
\begin{lemma}
	Taking the maximum in (\ref{eqn:FeedbackCapacity}) 
	over $\mathcal{P}^* \subset \mathcal{P}$,
	\begin{equation}
		\label{eqn:Lemma1}
		\max_{p_{X^n|Y^n}(x^n \given y^n) \in \mathcal{P}^{*}} 
		\left(
			\lim_{n \rightarrow \infty} \frac{1}{n} I(X^n \rightarrow Y^n)
		\right)
		= \cfb .
	\end{equation}
\end{lemma}
\begin{proof}
	The lemma follows from \cite[Thm. 1]{yin03}.
\end{proof}
%

It turns out that the feedback-capacity-achieving input distribution 
in $\mathcal{P}^*$ causes $Y^n$ to be a Markov chain (the reader may check; 
see also \cite{yin03,che05}).
That is,
\begin{equation}
	p_{Y_n|Y^{n-1}}(y_n \given y^{n-1}) = p_{Y_n|Y_{n-1}}(y_n \given y_{n-1}) .
\end{equation}
Using the following shorthand notation:
\begin{eqnarray}
	p_{L|B}^{(i)} & := & p_{X_i |Y_{i-1}}(\L \given \B) \\
	p_{L|U}^{(i)} & := & p_{X_i | Y_{i-1}}(\L \given \U) \\
	\bar{\alpha}^{(i)} & := & \alpha_\H (1-p_{\L|\U}^{(i)}) + \alpha_\L p_{\L|\U}^{(i)} ,
\end{eqnarray}
where the superscripts represent the time index,
the transition probability matrix for $Y$ at time $i$, $\mathbf{P}_Y^{(i)}$, is 
\begin{equation}
	\label{eqn:PYMatrix}
	\mathbf{P}_Y^{(i)} 
	=
	\left[
		\begin{array}{cc}
			1- \bar{\alpha}^{(i)} 
				& \bar{\alpha}^{(i)} \\
			\beta & 1-\beta
		\end{array}
	\right] ,
\end{equation}
with the first row and column corresponding to $U$, and the second row and column
corresponding to $B$. 


%

We now consider stationary distributions.
Let 
$\mathcal{P}^{**} \subset \mathcal{P}^*$ represent the distributions
that can be written with stationary $p_{X_i|Y_{i-1}}(x_i \given y_{i-1})$, i.e.,
with some time-independent distribution $p_{X|Y}$ such that 
\begin{equation}
	p_{X^n|Y^n}(x^n \given y^n) = \left(\prod_{k=2}^n p_{X|Y}(x_i \given y_{i-1})\right) p_{X_1}(x_1).
\end{equation}
Then:
\begin{lemma}
	\label{lem:Independent}
	Taking the maximum in (\ref{eqn:FeedbackCapacity}) 
	over $\mathcal{P}^{**} \subset \mathcal{P}^* \subset \mathcal{P}$,
		\begin{equation}
			\label{eqn:Lemma2}
			\max_{p_{X^n|Y^n}(x^n \given y^n) \in \mathcal{P}^{**}} 
			\left(
				\lim_{n \rightarrow \infty} \frac{1}{n} I(X^n \rightarrow Y^n)
			\right)
			= \ciid .
		\end{equation}
\end{lemma}
%
%
\begin{proof}
	We start by showing that
	$I(X^i; Y_i \given Y^{i-1})$ is independent of $p_{L|B}^{(k)}$ for all $k$.
	There is a feedback-capacity-achieving input distribution in $\mathcal{P}^*$ (from Lemma 1).
	Using this input distribution,
	\begin{eqnarray}
		\lefteqn{I(X^i; Y_i \given Y^{i-1})} & & \nonumber \\ 
		& = & H(Y_i \given Y^{i-1}) - H(Y_i \given Y_{i-1}, X^i) \\
		\label{eqn:DirectedTermSimplified}
		& = & H(Y_i \given Y_{i-1}) - H(Y_i \given Y_{i-1}, X_i) .
	\end{eqnarray}
	where (\ref{eqn:DirectedTermSimplified}) follows since (by definition)
	$Y_i$ is conditionally independent of $X^{i-1}$ given $Y_{i-1}$,
	and since $Y^i$ is first-order Markov. 
	Expanding (\ref{eqn:DirectedTermSimplified}),
	\begin{eqnarray}
		\nonumber
		\lefteqn{I(X^i; Y_i \given Y^{i-1}) = } & & \\ 
		\label{eqn:DirectedInfoSums}
		& & \sum_{y_{i-1}} p_{Y_{i-1}}(y_{i-1}) \sum_{x_i} p_{X_i|Y_{i-1}}(x_i \given y_{i-1}) \\
		\nonumber
		& & \cdot \sum_{y_i} p_{Y_i | Y_{i-1}, X_i}(y_i | y_{i-1}, x_i)
			\log \frac{p_{Y_i | Y_{i-1}, X_i}(y_i | y_{i-1}, x_i)}
			{p_{Y_i | Y_{i-1}}(y_i | y_{i-1})} .
	\end{eqnarray}
	From (\ref{eqn:PYMatrix}), 	
	$p_{Y_{i-1}}(y_{i-1})$ is calculated from parameters in 
	$\mathbf{P}_Y^{(i)}$ and the initial state,
	so $p_{Y_{i-1}}(y_{i-1})$ is independent of $p_{L|B}^{(k)}$ for all $k$.
	Further, everything under the last sum (over $y_i$) is independent of $p_{L|B}^{(k)}$, from 
	(\ref{eqn:PYMatrix}) and the definition of $p_{Y_i | Y_{i-1}, X_i}(y_i \given y_{i-1}, x_i)$.
	There remains the term $p_{X_i|Y_{i-1}}(x_i \given y_{i-1})$, which is dependent on
	$p_{L|B}^{(i-1)}$ when $y_{i-1} = B$.
	However, if $y_{i-1} = B$, then 
	\begin{eqnarray}
		\nonumber \lefteqn{\sum_{y_i} p_{Y_i | Y_{i-1}, X_i}(y_i \given B, x_i)
			\log \frac{p_{Y_i | Y_{i-1}, X_i}(y_i \given B, x_i)}
			{p_{Y_i | Y_{i-1}}(y_i \given B)}} & & \\
		\label{eqn:InnerSumZero}
		& = & \sum_{y_i} p_{Y_i | Y_{i-1}}(y_i \given B)
			\log \frac{p_{Y_i | Y_{i-1}}(y_i \given B)}
			{p_{Y_i | Y_{i-1}}(y_i \given B)} \\
			& = &  \sum_{y_i} p_{Y_i | Y_{i-1}}(y_i \given B) \log 1 \\
		& = & 0 ,
	\end{eqnarray}
	where (\ref{eqn:InnerSumZero}) follows since $y_i$ is independent of $x_i$ in state $B$.
	Thus, the entire expression is independent of $p_{L|B}^{(k)}$ for all $k$.
	Moreover, from (\ref{eqn:DirectedInfo}), directed information
	is independent of $p_{L|B}^{(k)}$ for all $k$.
	
	To prove (\ref{eqn:Lemma2}), distributions in $\mathcal{P}^{**}$ have
	$p_{L|U}^{(1)} = p_{L|U}^{(2)} = \ldots$, and 
	$p_{L|B}^{(1)} = p_{L|B}^{(2)} = \ldots$. 
	Since $I(X^i; Y_i \given Y^{i-1})$ is independent of $p_{L|B}^{(k)}$ for all $k$ (by the 
	preceding argument),
	we may set $p_{L|B}^{(k)} = p_{L|H}^{(k)}$ for all $k$, 
	without changing $I(X^i; Y_i \given Y^{i-1})$. Thus, inside $\mathcal{P}^{**}$,
	there exists a maximizing input 
	distribution that is independent for each channel use. By the definition
	of $\mathcal{P}^{**}$, that maximizing input distribution is iid, and there cannot exist
	an iid input distribution outside of $\mathcal{P}^{**}$. 
\end{proof}
%

Finally, we must show that feedback capacity is itself achieved by
a stationary input distribution.
To do so, we rely on \cite[Thm. 4]{che05}, which states that there is a feedback-capacity-achieving
input distribution in $\mathcal{P}^{**}$, as long as several technical conditions are satisfied.
Stating the conditions and proving that they hold requires restatement of 
definitions from \cite{che05}, so we give this result in the
appendix as Lemma 3.

Up to now, we have dealt only with feedback capacity.
We now return to the proof of Theorem \ref{thm:main}, where we relate these 
results to the regular capacity $C$.
\begin{proof}
%
	From Lemma 1, $\cfb$ is satisfied by an input distribution in $\mathcal{P}^*$.
	From Lemma 2, if we restrict ourselves to the stationary 
	input distributions $\mathcal{P}^{**}$ (where $\mathcal{P}^{**} \subset \mathcal{P}^*$), 
	then the feedback capacity is $\ciid$. From Lemma 3, the conditions of \cite[Thm. 4]{che05} are
	satisfied, which implies that there is a 
	feedback-capacity-achieving input distribution in $\mathcal{P}^{**}$.
	Therefore,
	\begin{equation}
		\cfb = \ciid .
	\end{equation}
	Considering the regular capacity $C$, $\cfb \geq C$, since the receiver has the option to ignore feedback; and $C \geq \ciid$,
	since an iid input distribution is a possible (feedback-free) input distribution. Thus,
	$\cfb = C = \ciid$.
\end{proof}

Finally,
if the input distribution is iid, then $Y^n$ is a Markov chain (see also the discussion after Lemma 1), and the mutual information rate can be expressed in closed form. Let 
$\mathcal{H}(p) = - p \log_2 p - (1-p) \log_2 (1-p)$ represent the binary
entropy function. In the iid input distribution, let $p_\L$ and $p_\H$ represent the probability of 
low and high concentration, respectively. Then
\begin{eqnarray}
	\lefteqn{\lim_{n \rightarrow \infty} \frac{1}{n} I(X;Y)} & & \nonumber \\
	& = & H(Y_n \given Y_{n-1}) - H(Y_n \given X_n,Y_{n-1}) \\
	& = & \frac{\mathcal{H}(\alpha_\H p_\H + \alpha_\L p_\L) 
		- p_\H \mathcal{H}(\alpha_\H) - p_\L \mathcal{H}(\alpha_\L)}
		{1 + (\alpha_\H p_\H + \alpha_\L p_\L)/\beta} .
\end{eqnarray}
Maximizing this expression with respect to $p_\L$ and $p_\H$  
(with appropriate constraints) gives the capacity.
It is straightforward to show that the largest possible value of the capacity is obtained in the limit 
$\alpha_L,\beta \to 0$ and $\alpha_H \to 1$; 
in this case the capacity is exactly $C=\log\phi\approx 0.694242$ (bits per time step), where 
$\phi=(1+\sqrt{5})/2$; this capacity is achieved when $p_H=\phi-1\approx0.381966$.  

\section{Acknowledgments}
The authors thank Robin Snyder and Marshall Leitman for their comments, 
and Toby Berger for giving us a copy of \cite{yin-unpublished}.

\appendix

We start by defining strong irreducibility and strong aperiodicity for $Y^n$,
assuming that the input distribution is in $\mathcal{P}^*$ (i.e., $Y^n$ is a Markov chain).
Recalling (\ref{eqn:LowConcentrationOutput})-(\ref{eqn:HighConcentrationOutput}),
let $\hat{\mathbf{P}} = [\hat{P}_{ij}]$ represent a $2 \times 2$ $\{0,1\}$ matrix with elements
\begin{equation}
	\hat{P}_{ij} =
	\left\{
		\begin{array}{cl}
			1, & \min_{k \in \{\L,\H\}} P_{Y|X=k,ij} > 0 \\
			0, & \mathrm{otherwise} ,
		\end{array}
	\right.
\end{equation}
and for positive integers 
$\ell$, let $\hat{P}_{ij}^\ell$ represent the $i,j$th element of $\hat{\mathbf{P}}^\ell$.
Further, for the $i$th diagonal element of the $\ell$th matrix power $\hat{P}_{ii}^\ell$, 
let $\mathcal{D}_i$ contain the set of integers $\ell$ such that 
$\hat{P}_{ii}^\ell \neq 0$. 
Then: 
\begin{itemize}
	\item $Y^n$ is strongly irreducible if, for each pair $i,j$, there exists 
an integer $\ell > 0$
such that $\hat{P}_{ij}^\ell \neq 0$; and
	\item If $Y^n$ is strongly irreducible, it is also strongly aperiodic if,
	for all $i$, the greatest common divisor of $\mathcal{D}_i$ is 1.
\end{itemize}
These conditions are described in terms of graphs in \cite{che05}, but our description
is equivalent.

%
\begin{lemma}
	If $0 < \alpha_\L, \alpha_\H, \beta < 1$,  the conditions of \cite[Thm. 4]{che05}
	are satisfied, namely:
	\begin{enumerate}
		\item $Y^n$
		is strongly irreducible and strongly aperiodic.
		\item 
		For $i \in \{\U,\B\}$, let 
\begin{eqnarray}
	\lefteqn{\mathbf{R}_i = } & & \\
	& & \nonumber \left[ 
		\begin{array}{cc}
			p_{Y_t|X_t,Y_{t+1}}(\B \given \L, i) & p_{Y_t|X_t,Y_{t+1}}(\B \given \H, i) \\
			p_{Y_t|X_t,Y_{t+1}}(\U \given \L, i) & p_{Y_t|X_t,Y_{t+1}}(\U \given \H, i)
		\end{array}
	\right] ,
\end{eqnarray} 
and  
let $I(p,\mathbf{R}_i) = I(Y_1;X_1 \given Y_0 = i)$, where the input distribution is 
$p \in \mathcal{P}^*$, and the input-output probabilities are given by $\mathbf{R}_i$.
Then (reiterating \cite[Defn. 6]{che05}) for the set of possible input distributions in $\mathcal{P}^*$, 
and for all $i \in \{\U,\B\}$, 
there exists a subset $\tilde{\mathcal{P}}^*$ satisfying%
\begin{enumerate}
	\item $\{\mathbf{R}_i p : p \in \mathcal{P}^*\} = \{\mathbf{R}_i p : p \in \tilde{\mathcal{P}}^*\}$. 
	\item For any $r \in \{\mathbf{R}_ip : p \in \mathcal{P}^*]\}$,
	\begin{eqnarray}
	\label{eqn:Defn6Part2}
		\left\{ \arg \max_{p:p\in\mathcal{P}^* \atop \mathbf{R}_i p=r} I(p,\mathbf{R}_i) \right\} 
		\cap 
		\left\{ \arg \max_{p:p\in\tilde{\mathcal{P}}^* \atop \mathbf{R}_i p=r} I(p,\mathbf{R}_i) \right\}  & &
		\nonumber \\
		 \neq \emptyset & &
	\end{eqnarray}
	\item There exists a positive constant $\lambda$ such that
	\begin{equation}
		\frac{\partial I(p_2,\mathbf{R}_i)}{\partial \ell} - \frac{\partial I(p_1,\mathbf{R}_i)}{\partial \ell}
		\leq
		- \lambda || p_2 - p_1 ||
	\end{equation}
	for any nonidentical $p_1,p_2 \in \tilde{\mathcal{P}}^*$, where $\ell$ is in the direction from 
	$p_1$ to $p_2$, and the norm is the Euclidean vector norm.
\end{enumerate}
\end{enumerate}
\end{lemma}

\begin{proof}
	To prove the first part of the lemma, if $0 < \alpha_\L,\alpha_\H,\beta < 1$, then
	$\hat{\mathbf{P}}$ is an all-one matrix, so
	$Y^n$ is strongly irreducible (with $\ell = 1$); further, since the positive powers of an all-one
	matrix can never have zero elements, $\mathcal{D}_i$ contains all positive integers from 1 to $n$,
	whose greatest common divisor is 1, so $Y^n$ is strongly aperiodic.
	
	To prove the second part of the lemma,
	we first show that the definition is satisfied for $\mathbf{R}_\B$, given by
\begin{equation}
	\label{eqn:QB}
	\mathbf{R}_\B = \left[ \begin{array}{cc} 1-\beta & 1-\beta \\ \beta & \beta \end{array}\right] .
\end{equation}
We choose the subset 
$\tilde{\mathcal{P}}^*$ to consist of a single point $p \in \mathcal{P}^*$ (it can be any point, as all
points give the same result).
The columns of $\mathbf{R}_\B$ are
identical, since the output is not dependent
on the input in state $\B$.
Then for every $p\in\mathcal{P}^*$, 
\begin{equation}
	\mathbf{R}_\B p = 
		\left[ \begin{array}{cc} 1-\beta & 1-\beta \\ \beta & \beta \end{array}\right] 
		\left[ \begin{array}{c} p_{\L|\B} \\ p_{\H|\B} \end{array}\right] 
	= \left[ \begin{array}{c} 1-\beta \\ \beta \end{array}\right] .
\end{equation}
This is also true of the single point in $\tilde{\mathcal{P}}^*$,
so condition 1 is satisfied.
Similarly, by inspection of (\ref{eqn:QB}), when $Y_0 = \B$, 
the output $Y_1$ is not dependent on the input $X_1$, so 
$I(p,\mathbf{R}_\B) = 0$ for all $p \in \mathcal{P}$. 
Since all $p \in \mathcal{P}^*$ ``maximize'' $I(p,\mathbf{R}_\B)$
and have identical values of $\mathbf{R}p$ (including the single point in $\tilde{\mathcal{P}}^*$),
then the single point $p \in \tilde{\mathcal{P}}^*$ is always in both sets, and the intersection
(\ref{eqn:Defn6Part2})
is nonempty; so condition 2 is satisfied. There is only one point in 
$\tilde{\mathcal{P}}^*$, so condition 3 is satisfied trivially.

Now we show the conditions are satisfied for $\mathbf{R}_\U$, given by
\begin{equation}
	\mathbf{R}_\U = \left[ \begin{array}{cc} \alpha_\L & \alpha_\H \\ 1-\alpha_\L & 1-\alpha_\H \end{array}\right] .
\end{equation}
There are two possibilities. First, suppose $\alpha_\L = \alpha_\H$, so that $\mathbf{R}_\U$ has the same
form as $\mathbf{R}_\B$; then $\mathbf{R}_\U$ satisfies the conditions by the same argument that we gave above.
Second, suppose $\alpha_\L \neq \alpha_\H$; then $\mathbf{R}_\U$ has rank 2, so 
by \cite[Lem. 6]{che05}, $\mathbf{R}_\U$ satisfies the conditions.
\end{proof}
Closely related 
results on feedback capacity of binary channels were given in \cite{yin-unpublished} (unfortunately, unpublished). 

\bibliographystyle{ieeetr} 
\bibliography{infotheory,Dicty,stoch_chem,neuroscience,MCell,PJT,isit2013}

\end{document}